\documentclass[journal]{IEEEtran}
\ifCLASSINFOpdf
\else
   \usepackage[dvips]{graphicx}
\fi
\usepackage{url}
\usepackage{graphicx}
\usepackage{subfig}
\usepackage{float}
\usepackage{amsmath,amssymb}
\usepackage{booktabs} 
\newtheorem{proposition}{\textbf{Proposition}}
\newtheorem{proof}{Proof}

\begin{document}

\title{Optimal Anchor Deployment and Topology Design for Large-Scale AUV Navigation}

\author{Wei Huang, \IEEEmembership{Member, IEEE}, Junpeng Lu, Tianhe Xu, Jianxu Shu, Hao Zhang, \IEEEmembership{Senior Member, IEEE}, and Kaitao Meng, \IEEEmembership{Member, IEEE}, Yanan Wu
\thanks{Submitted XX 2025. This work was financially supported in part by National Key Research and Development Program of China (2024YFB3909701), in part by National Natural Science Foundation of China (42404001,62271459), in part by the Stable Supporting Fund of Acoustic Science and Technology Laboratory (JCKYS2025SSJS008), in part by Shandong Provincial Natural Science Foundation (ZR2023QF128).}
\thanks{Wei Huang, Junpeng Lu, Hao Zhang and Yanan Wu are with the Faculty of Information Science and Engineering, Ocean University of China, Qingdao 266100, China (e-mail:hw@ouc.edu.cn,lujunpeng@stu.ouc.edu.cn,zhanghao@ouc.edu.cn,wuya-nan@ouc.edu.cn).}
\thanks{Tianhe Xu and Jianxu Shu are with the School of Space Science and Technology, Shandong University (Weihai), Weihai 264209, China (e-mail: thxu@sdu.edu.cn, shujx@mail.sdu.edu.cn).}
\thanks{Kaitao Meng is with the Department of Electrical and Electronic Engineering, The University of Manchester, Manchester, UK (emails: kaitao.meng@manchester.ac.uk).}
\thanks{Corresponding author: Kaitao Meng.}}

\markboth{IEEE, Vol. XX, No. XX, XX 2025}
{Huang \MakeLowercase{\textit{et al.}}: Optimal Deployment of Anchor Cluster for Large-scale AUV Navigation}
\maketitle

\begin{abstract}
Seafloor acoustic anchors are an important component of AUV navigation, providing absolute updates that correct inertial dead-reckoning. Unlike terrestrial positioning systems, the deployment of underwater anchor nodes is usually sparse due to the uneven distribution of underwater users, as well as the high economic cost and difficult maintenance of underwater equipment. These anchor nodes lack satellite coverage and cannot form ubiquitous backhaul as terrestrial nodes do. In this paper, we investigate the optimal anchor deployment topology to provide high-quality AUV navigation and positioning services. We first analyze the possible deployment mode in large-scale underwater navigation system, and formulate a topology optimization for underwater anchor node deployment. Then, we derive a scaling law about the influence of anchors in each cluster on the navigation performance within a given area and demonstrate a service area coverage condition with a high probability of reaching the destination. Finally, the optimization performance is evaluated through experimental results.
\end{abstract}


\IEEEpeerreviewmaketitle

\section{Introduction}
\IEEEPARstart{U}{nderwater} vehicles have played an important role in many fields such as underwater data collection, underwater rescue, and exploration of marine resources. Due to severe attenuation of electromagnetic waves underwater, underwater vehicles cannot be located through global navigation satellite systems (GNSS) and usually rely on dead-reckoning for real-time position estimation \cite{Erol2011LocSurvey}. However, without external information to assist in position correction, the dead-reckoning errors will accumulate continuously as time goes by \cite{Wang2020Navigation}. Although instruments such as Gravimeters and Doppler velocity logs (DVLs) have achieved significant results in assisting underwater navigation and positioning in recent years, vehicles still rely on the absolute position correction of underwater long baseline acoustic positioning systems \cite{Luo2021LocSurvey}.

\indent Unlike satellite systems, the communication coverage of anchor node is quite limited compared with GNSS. At the same time, due to the fact that anchor nodes are economically expensive and difficult to maintain, the anchor node cannot be deployed ubiquitously, unlike terrestrial systems \cite{Hu2024SparseeNode}. Thus, how to efficiently utilize a limited number of underwater anchors for providing good navigation and positioning services has attracted more and more attention.

\indent For a given target to be located, the optimal deployment of reference anchor nodes is an important way to improve the localization accuracy performance, as it decides the node geometric dilution of precision (GDOP) \cite{Huang2024Deployment,Meng2025TWC}. In this perspective, many works have been done in optimal reference nodes selection to improve the performance of target positioning. Zhao et al. proposed a D-optimality criterion, which maximize the determinant of Fisher information matrix (FIM), so as to minimize the Cram$\acute{e}$r lower bound (CRLB) \cite{Zhao2013OptSensor}. In \cite{Rui2014OptReceiver}, the authors proposed an optimal node deployment strategy according to an A-optimality criterion that minimizing the trace of inverse FIM for time-of-arrival (TOA) localization. In \cite{Wang2021APOR} and \cite{Xue2022JOG}, the authors adopted ships to provide virtual references along circular paths to locate an underwater anchor, but they did not analyze how to determine the optimal horizontal distance between the reference position and the target. These works are suitable for circumstances where there are sufficient reference nodes within the communication coverage of the target to be located. However, as mentioned above, there may be only a few numbers of anchor nodes due to economic and maintenance reasons, thus the freedom of selecting reference nodes is limited.
\vspace{-5mm}
\begin{figure}[!htbp]
	\centering
	\subfloat[]{\includegraphics[width=0.23\linewidth]{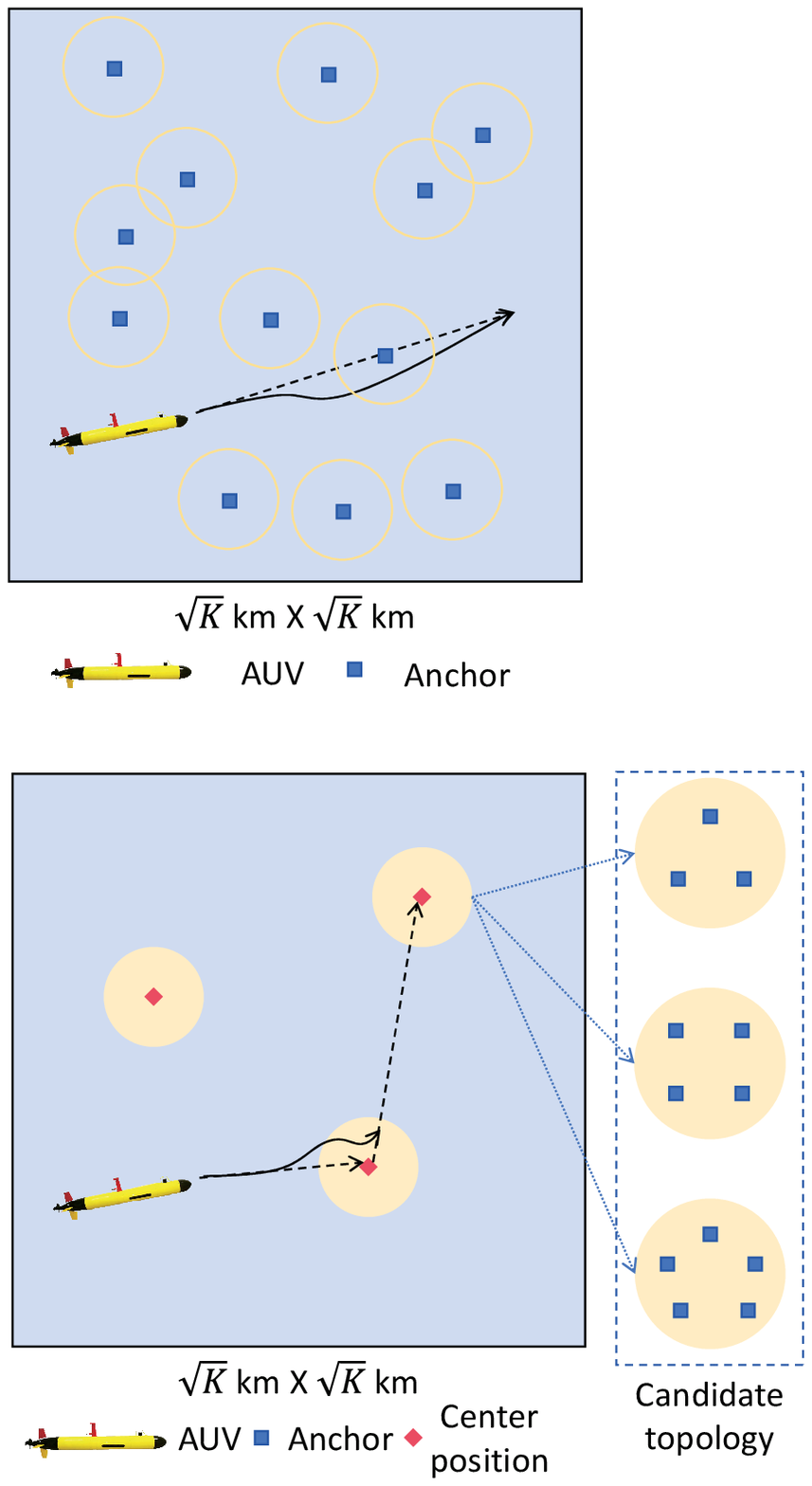}\label{fig01a}}
	\subfloat[]{\includegraphics[width=0.22\linewidth]{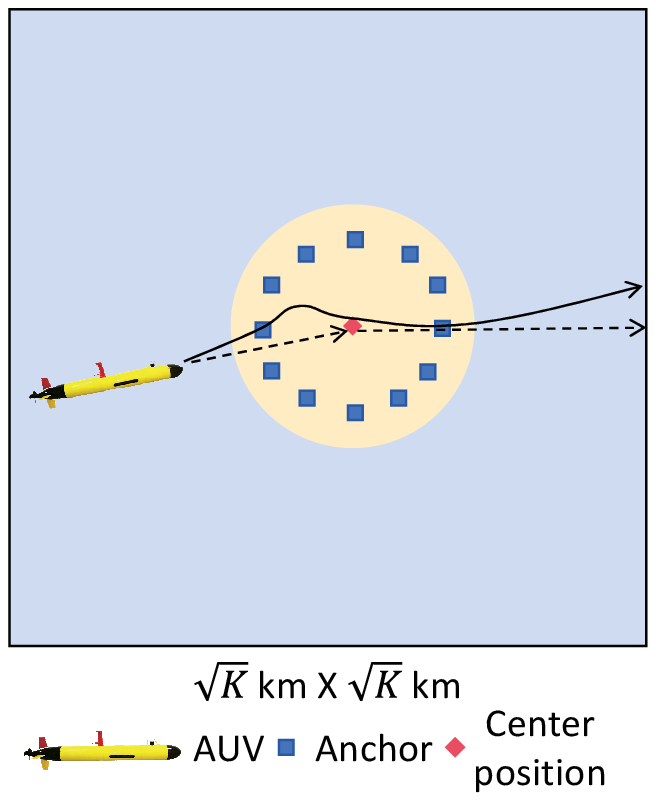}\label{fig01b}}
	\subfloat[]{\includegraphics[width=0.3\linewidth]{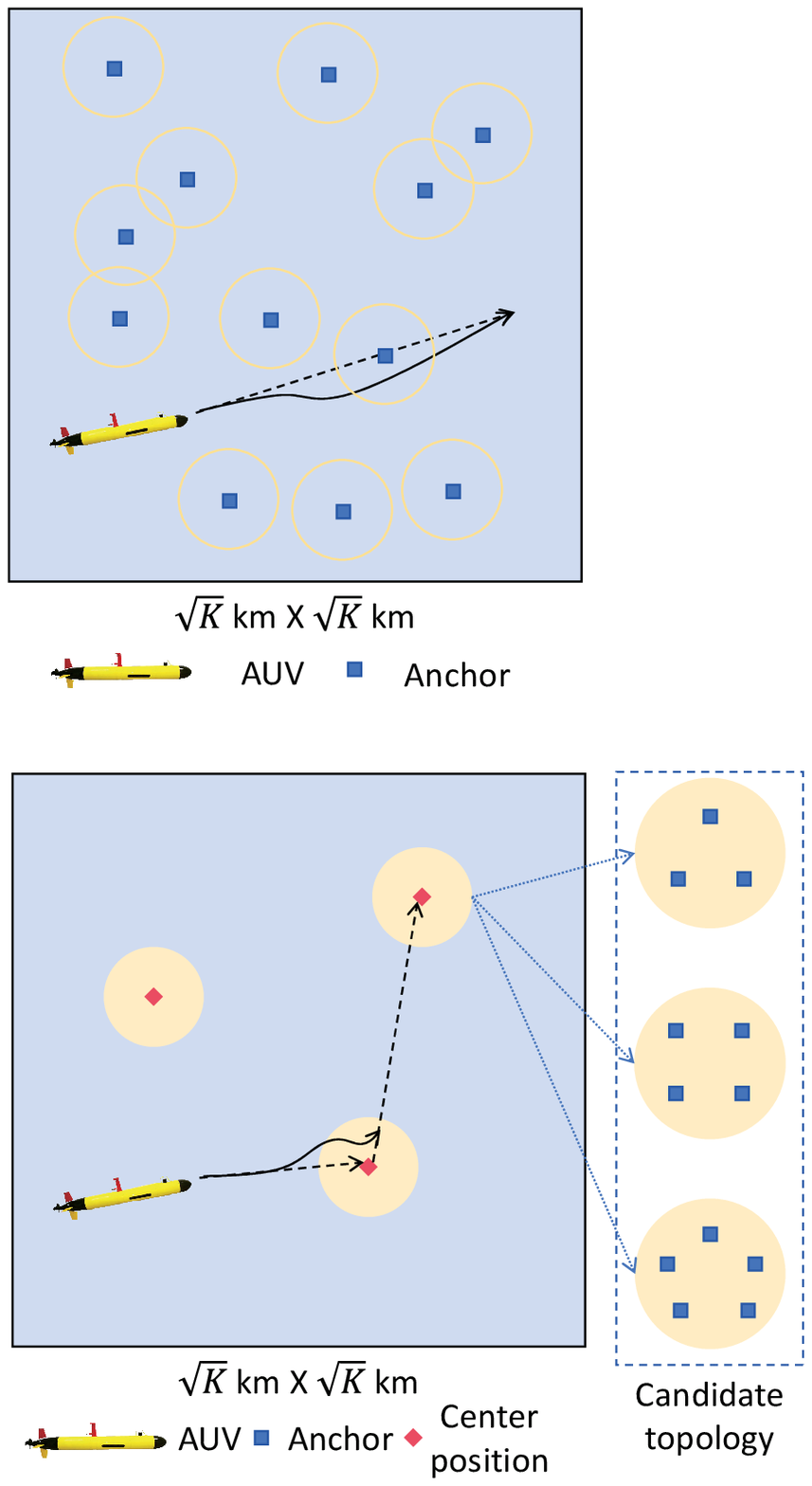}\label{fig01c}}
	\caption{Deployment modes of anchors. (a)Random sparse distribution of anchor nodes. (b)Centralized deployment of Anchor Nodes. (c)Distributed deployment of Anchor Node Sets. }
	\label{fig01}
\end{figure}
\vspace{-5mm}

\indent Considering the AUV navigation in a large-scale region with a given number of anchor nodes as shown in Fig.~\ref{fig01}, if they are randomly distributed in Fig.~\ref{fig01a}, AUV will not be able to simultaneously receive acoustic signal observation information that meets the positioning requirements. Therefore, a more reasonable approach is to cluster a small number of anchors into a group for centralized deployment, while the issue that should be considered is how to deploy the anchor cluster as shown in Fig.~\ref{fig01b} and Fig.~\ref{fig01c}. In this paper, we study the optimal anchor deployment topology for large-scale AUV navigation, aiming to minimize the expected error of AUV's position estimation in the whole area. The contributions of this paper are summarized as follows.

\begin{itemize}
\item Based on practical requirements for deploying sparse nodes in large-scale ocean navigation systems, we simultaneously study the positioning accuracy of anchor cluster coverage and non coverage areas, and establish an optimized deployment model for anchor clusters.
\item To improve the performance of the positioning system, we analyze the seafloor anchor node positioning scene considering the Snell effect, and optimize the deployment of reference nodes by minimizing the CRLB.
\item We discovered the service coverage conditions that ensure stable navigation of underwater vehicles and derived a scaling law about the influence of anchors in each cluster on the localization performance within a given area.
\end{itemize}

The rest of this paper is organized as follows. Section \ref{sec2} describes the system model. The problem formulation and optimization solution is given in section \ref{sec3}. Numerical results and discussions are given in section \ref{sec4}.

\vspace{-3mm}
\section{System Model}\label{sec2}
\vspace{-5mm}
\begin{figure}[!htbp]
	\centering
	\begin{minipage}[h]{0.23\textwidth}
		\centering
		\includegraphics[width=\linewidth]{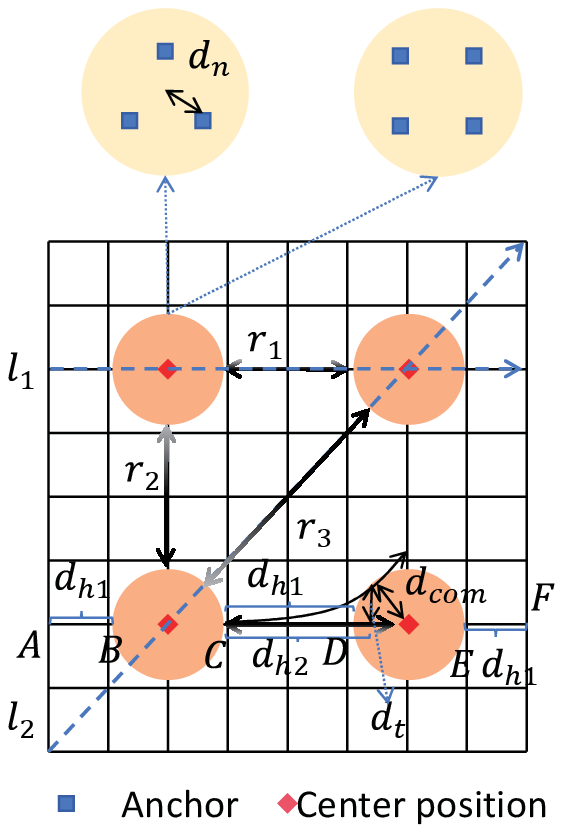}%
		\caption{Uniform deployment for fair service in a $\sqrt{k} \times \sqrt{k}$ area.}
		\label{fig02}
	\end{minipage}
	\begin{minipage}[h]{0.24\textwidth}
		\centering
		\includegraphics[width=\linewidth]{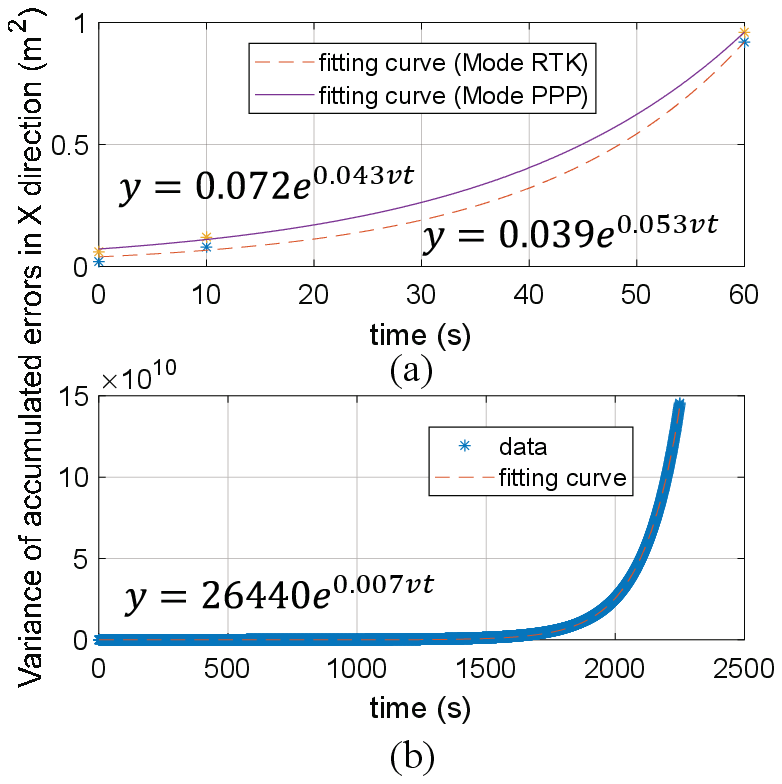}%
		\caption{The variance of X-direction position error over time in exponential fitting. (a) Product factory testing results. (b) Real-world experimental results. (Data will be open at www.weiwilliamhuang.cn)}
		\label{fig03}
	\end{minipage}
\end{figure}

\indent Considering an AUV sailing across an area of $\sqrt{k}$ km $\times$ $\sqrt{k}$ km, the navigation process of the AUV is divided into two stages, one is the autonomous navigation stage outside the coverage range of the anchor cluster, and the other is the acoustic navigation stage within the anchor cluster as shown in Fig.~\ref{fig02}. In this paper, we aim to minimize the expected positioning error of AUVs passing through the area by optimizing the deployment of underwater anchor nodes.

\subsection{Navigation Model}

When there is no acoustic assisted positioning, AUVs rely on inertial navigation to obtain real-time position. To reduce the accumulation of positioning errors, AUV usually passes through the coverage area of the anchor clusters as much as possible during navigation. However, between two clusters of anchors, the position error of AUV accumulates with the time, and few studies have analyzed the distribution characteristics of inertial navigation errors from a statistical perspective. To guide subsequent research on the impact of inertial navigation errors on node deployment strategies, we give an exponential functional model to describe the variance of accumulated errors. Fig.~\ref{fig03} shows the error accumulation of a measured inertial navigation system in horizontal direction with IMU-ISA-100C produced by NOVATEL company \cite{IMU2025}.

Let the moment when the AUV leaves a certain anchor cluster be the starting time, and after time $t$, the position of the inertial navigation system will be $\hat{\boldsymbol{p}}_t = [\hat{x}_t,\hat{y}_t,\hat{z}_t]^T$. According to the kinematic equation, there is:
\begin{equation}
\hat{\boldsymbol{p}}_t = \boldsymbol{p}_0 + \int_{0}^{t} \boldsymbol{v}(t) dt, \label{eq1}
\end{equation}
where $\boldsymbol{p}_0 = [x_0,y_0,z_0]^T$ is the start position leaving the first anchor cluster, $t$ is the sailing time, $\boldsymbol{v}=[v_x, v_y, v_z]^T$ is the speed vector, $\Delta \boldsymbol{p} = \int_{0}^{t} \boldsymbol{v}(t) dt$ is the displacement.

The motion of AUV is nonlinear and difficult to accurately describe its trajectory. To simplify the problem, we discretize the time line. When the time slot $T$ is small enough, it can be considered that the AUV moves linearly with uniform speed or acceleration during the time slot. Thus, for time moment t, the position of the inertial navigation system can be described as $\hat{\boldsymbol{p}}_t = \boldsymbol{p}_0 + \sum_{t=1}^{t} \boldsymbol{\bar{v}}_tT$.

Let the real position of AUV be $\boldsymbol{p}_t = [x_t,y_t,z_t]^T$, we have $\boldsymbol{p}_t = \hat{\boldsymbol{p}}_t + \boldsymbol{e}_{nv}(t)$, where $\boldsymbol{e}_{nv}(t) = [e_{t,x},e_{t,y},e_{t,z}]^T$ denotes the accumulated navigation error, and $e_{t,x},e_{t,y},e_{t,z} \sim N(0,\sigma_t^2)$. Especially, when the AUV just leaves an anchor cluster, the real position will be $\boldsymbol{p}_0 = \hat{\boldsymbol{p}}_0 + \boldsymbol{e}_n(t=0)$, where $\boldsymbol{e}_n(t=0) = [e_{0,x},e_{0,y},e_{0,z}]^T$ is the initial random navigation error, the elements of which all follow Gaussian distribution with $e_{0,x},e_{0,y},e_{0,z} \sim N(0,\sigma_0^2)$.

From the experimental results shown in Fig.~\ref{fig03}, the variance of location error brought by navigation follows exponential divergence:
\begin{equation}
\boldsymbol{\sigma}_t^2 = \boldsymbol{\sigma}_0^2 + \beta_1\mathrm{e}^{\beta_2\Delta \boldsymbol{p}_t}, \label{eq2}
\end{equation}
where $\beta_1$ and $\beta_2$ are diffusion coefficients, and $\mathrm{e}$ is the natural base number. 

Based on this exponential diffusion model, the position error still follows a Gaussian distribution with zero mean value, so the position $\hat{\boldsymbol{p}}_t$ estimated by the navigation system  is an unbiased estimation of the real location $\boldsymbol{p}_t$ because $E[\hat{\boldsymbol{p}}_t] = \boldsymbol{p}_t$. Then, the covariance matrix of the position error during pure navigation area outside the anchor cluster is
\begin{equation}
\boldsymbol{C}_{nv} = \mathbb{E}[(\hat{\boldsymbol{p}}_t - \mathbb{E}[\hat{\boldsymbol{p}}_t])(\hat{\boldsymbol{p}}_t - \mathbb{E}[\hat{\boldsymbol{p}}_t])^T] = \mathbb{E}[e_{nv}^2(t)],  \label{eq3}
\end{equation}
which can be expressed as $\boldsymbol{C}_{nv} = diag[\sigma^2_{t,x},\sigma^2_{t,y},\sigma^2_{t,z}]$, and $\sigma_{t,x},\sigma_{t,y},\sigma_{t,z}$ are the components of $\sigma_{t}$ in $x,y,z$ directions, respectively. We quantify navigation accuracy by:
\begin{equation}
\begin{split}
	e_{nv} &= \mathbb{E}\left[tr(\boldsymbol{C}_{nv})\right]=\frac{1}{N}\sum_{n=1}^{N}(\sigma_{nT,x}^2+\sigma_{nT,y}^2+\sigma_{nT,z}^2).  
\end{split}\label{eq4}
\end{equation}
where $N$ is the number of sampling points for positioning.

\vspace{-5mm}
\subsection{Acoustic Positioning Model}
\vspace{-5mm}
\begin{figure}[!htbp]
	\centering
	\includegraphics[width=0.8\linewidth]{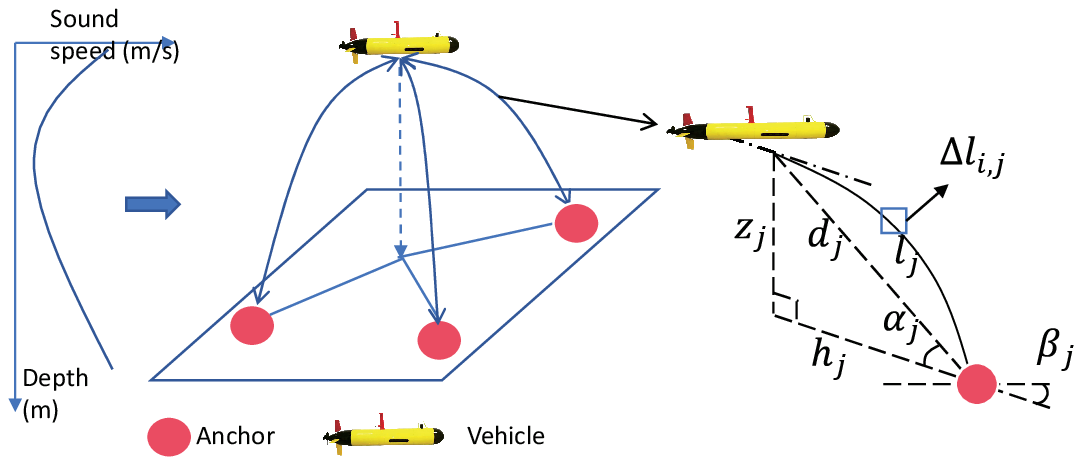}%
	\caption{Projection mechanism with signal refractive effect. Between target and anchor $j$, the elevation angle is $\alpha_j$, the azimuth angle is $\beta_j$, the vertical distance is $z_j$, the horizontal signal propagation distance is $h_j$, the real signal propagation path is $l_j$, and $\Delta l_{i,j}$ represents the signal propagation path at the $i$th depth layer.}
	\label{fig04}
\end{figure}
\vspace{-4mm}

Different from navigation with gradually accumulating errors, underwater acoustic positioning achieves instantaneous high-precision positioning through multi-point ranging, which can correct navigation position errors. Since the anchors are all deployed at the water bottom, at least three anchors are required for 3-D trilateration as shown in Fig.~\ref{fig04}. However, due to the non-uniformly distributed sound speed, there are signal refraction phenomena governed by Snell’s law. The equivalent line-of-sight (LOS) propagation distance needs to be corrected through ray tracing theory.

\indent Let the sound speed profile be $\boldsymbol{s} = [s_1,s_2,...,s_i],i=1,2,...,I$, the elevation angle be $\alpha_j$, the azimuth angle be $\beta_j$, the vertical distance be $z_j$, the horizontal signal propagation distance be $h_j$, the LOS distance be $d_j$, the real signal propagation path is $l_j$, and $\Delta l_{i,j}$ represents the signal propagation path at the $i$th depth layer, according to \cite{Huang2024Opt}, when the measurement error of propagation time is proportional to the real signal propagation distance and follows a Gaussian distribution, that is $n_{l,i,j}\sim N(0,\sigma_{l,i,j}^2)$ and $\sigma_{l,i,j} = \gamma \Delta l_{i,j}$ with $\gamma$ being a scale factor, the LOS measurement error between reference anchor j and the target will also follows a Gaussian distribution with $n_{d,j}\sim N\left(0,\sigma^2_{d,j}\right)$ where $\sigma^2_{d,j} \approx \sum_{i=2}^I \frac{\gamma^2 s_{i-1}^2}{s_1^2 - s_{i-1}^2\cos^2\alpha_{j}}$.

\indent Without loss of generality, assume there are J reference nodes that are within the communication range of the target, the anchor node measurement covariance matrix is given by
\begin{equation}
	\boldsymbol{C}_p = diag[\sigma^2_{d,1},\sigma^2_{d,2},...,\sigma^2_{d,J}]. \label{eq5}
\end{equation}
\indent The Jacobian matrix of the J measurements can be
\begin{equation}
	\boldsymbol{\mathcal{J}}_0 = \begin{bmatrix}
		\boldsymbol{a},\boldsymbol{b},\boldsymbol{c}
	\end{bmatrix},\label{eq6}
\end{equation}
where $\boldsymbol{a} = \left[\cos\alpha_1 \cos\beta_1,\cos\alpha_2 \cos\beta_2,...,\cos\alpha_J \cos\beta_J\right]^T$, $\boldsymbol{b} = \left[\cos\alpha_1 \sin\beta_1,\cos\alpha_2 \sin\beta_2,...,\cos\alpha_J \sin\beta_J\right]^T$, and $\boldsymbol{c} = \left[\sin\alpha_1,\sin\alpha_2,...,\sin\alpha_J\right]^T$. Then, the FIM will be:

\begin{equation}
	\begin{split}
		\boldsymbol{\Phi} &= \left[\boldsymbol{a} \quad \boldsymbol{b} \quad \boldsymbol{c}\right]^T \boldsymbol{C}_p^{-1}\left[\boldsymbol{a} \quad \boldsymbol{b} \quad \boldsymbol{c}\right].
	\end{split}\label{eq8}
\end{equation}

According to the A-optimization criterion that minimizing the trace of the inverse FIM in \cite{Xu2017Opt}, the CRLB will be
\begin{equation}
	\rm{CRLB} = \boldsymbol{\Phi}^{-1},\label{eq9}
\end{equation}
which reach the minimum value when the elevation angles $\alpha_1 = \alpha_2 =...=\alpha_J = \alpha $.

Our goal is to minimize the expected positioning error of underwater vehicles, so the positioning performance evaluation error is the expected value of CRLB:
\begin{equation}
	e_{p} = \mathbb{E}[\boldsymbol{\Phi}^{-1}]. \label{eq10}
\end{equation}

\subsection{Hybrid Model of Acoustic Positioning and Navigation}


\indent We aim to minimize the positioning error of the AUV during navigation throughout the entire region, for area covered by anchors, more anchors in one cluster may be helpful to improve positioning accuracy, which will reduce the number of anchor clusters. While for area uncovered by anchors, the cluster spacing is expected to be short, which requires reducing the number of nodes within the cluster. Therefore, there is a contradiction in the impact of cluster node settings on the localization performance of the two situations. 

Considering both stages simultaneously, we establish a hybrid model of acoustic positioning and navigation:
\begin{equation}
\begin{split}
	\rm{Opt} &=  \min ((1-\lambda_1) \mathbb{E}[tr(\boldsymbol{C}_t)] \\
	&+ \lambda_1 (\lambda_2 \mathbb{E}[\rm{CRLB}] +(1-\lambda_2)\mathbb{E}[tr(\boldsymbol{C}_t)]),
\end{split}\label{eq11}
\end{equation}
where $\lambda_1$ , $\lambda_2$ weight the uncovered/covered portions and the navigation/positioning trade-off within coverage, respectively \footnote{The underwater vehicle is assumed to travel straightly between adjacent anchor clusters.}. In the anchor coverage area, if the expected positioning error is much lower than the expected navigation error, the optimization problem can be simplified as
\begin{equation}
\rm{Opt} =  \min (\lambda_1 \mathbb{E}[\rm{CRLB}] + (1-\lambda_1) \mathbb{E}[tr(\boldsymbol{C}_t)]), \label{eq12}
\end{equation}
This situation is prone to occur because without the assistance of other sensors, navigation errors will rapidly increase, and with the help of acoustic positioning, the navigation error can be reset.

\vspace{-3mm}
\section{Optimal Anchor Cluster Deployment}\label{sec3}
\subsection{Anchor Deployment Analysis}
\subsubsection{Problem Formulation}
For a target area as shown in Fig.~\ref{fig02}, assume the total number of anchors to be deployed is $N_{ta}$, and the number of anchors in each cluster is $N_{ca}$, so the number of anchor clusters will be $N_c  = \lfloor N_{ta}/N_{ca}\rfloor$, where $\lfloor \cdot \rfloor$ means rounding down. In order to provide fair services, the anchor clusters are evenly distributed, so the number of anchor clusters in each direction is $\lfloor \sqrt{N_c} \rfloor$. Therefore, the distance between adjacent clusters is $d_c = \frac{\sqrt{k}}{\lfloor \sqrt{N_c} \rfloor}$.

For a given number of anchors in each cluster $N_{ca}$, the average positioning error of all positions covered by anchors can be solved through traversal method, and there will be:
\begin{equation}
	\mathbb{E}[\rm{CRLB}] = Q(N_{ca}), \label{eq13}
\end{equation}
which means the average positioning error will be a function of $N_{ca}$ because the topology shape and coverage range of anchors are different as $N_{ca}$ changes.

In real applications, AUVs are commonly equipped with pressure sensors, which can maintain stable depth information, so navigation positioning errors can be simplified as $tr(\boldsymbol{C}_t) =  (\sigma^2_{t,x} + \sigma^2_{t,y})$. Then, the optimization problem can be derived as:
\begin{equation}
	\begin{split}
	\rm{Opt} &=  \min (\lambda_1 Q_{N_{ca}} + (1-\lambda_1)\frac{1}{N}\sum_{n=1}^N (\sigma_{0,x}^2 +\sigma_{0,y}^2 \\
	&+ \beta_1\mathrm{e}^{\beta_2\Delta p_{nT,x}} + \beta_1\mathrm{e}^{\beta_2\Delta p_{nT,y}})), \label{eq14}
	\end{split}
\end{equation}
where $\Delta p_{nT,x}$ and $\Delta p_{nT,y}$ form three combinations that $\Delta p_{nT,x} = d_h, \Delta p_{nT,y} = 0$, $\Delta p_{nT,x} = 0, \Delta p_{nT,y} =  d_h$ and $\Delta p_{nT,x} = d_h, \Delta p_{nT,y} =  d_h$. $d_h \approx d_c - 2d_{com}$, where $d_{com}$ is approximately the localization coverage of a cluster. These three combinations are corresponding to path $r_1$, $r_2$, and $r_3$ in Fig.~\ref{fig02}, respectively. 

\subsubsection{Condition of Stable Navigation}
To make sure that the AUV can successfully arrive at the neighbor anchor cluster, the distance between adjacent anchor clusters $d_c$ can not be too large, so that with a given total number of anchors, the maximum coverage area $k$ should be less than a threshold value $k_{max}$. Concretely, according to Fig.~\ref{fig02}, the minimum navigation distance between anchor clusters is $d_{h1}$. $d_{h2}$ and $d_{t}$ are two distances that describe the AUV deviation from anchor cluster, which means when the AUV travels a distance of $d_{h2}$, navigation errors prevent it from reaching the next anchor cluster. Based on the above distance definition we can summarize that
\setcounter{proposition}{0} 
\begin{proposition}
	If the AUV can successfully sail to neighbor anchors in X or Y direction, there must be
	\begin{equation}
		\begin{aligned}
			d_{com}^2 - \frac{d_{com}^4}{(d_{com}+d_{h1})^2} - (\beta_1+1) e^{\beta_2 d_{h1}} -2\sigma_0^2 > 0.
		\end{aligned}\label{eq15}
	\end{equation}
\end{proposition}
where the navigation distance $d_{h1} = \frac{\sqrt{k}-2d_{com}N_{ch}}{N_{ch}+1}$, and the number of cluster in X or Y direction $N_{ch} = \lfloor \sqrt{\lfloor \frac{N_{ta}}{N_{ca}} \rfloor} \rfloor$.
\begin{proof}
	If the sailing distance of AUV is $d_{h2}$ as shown in Fig.~\ref{fig02}, there is
	\begin{equation}
		2\sigma_0^2 + (\beta_1+1) e^{\beta_2d_{h2}} < d_{t,max}^2,\label{eq16}
	\end{equation}
	According to the Pythagorean theorem and the triangle approximation rule, there is 
	\begin{equation}
		\frac{d_{com}^2 - d_{t,max}^2}{d_{com}^2} > \frac{d_{com}^2}{(d_{com}+d_{h1})^2}.\label{eq17}
	\end{equation}
	Considering $d_{h2} > d_{h1}$, there will be
	\begin{equation}
		d_{com}^2 - \frac{d_{com}^4}{(d_{com}+d_{h1})^2} > 2\sigma_0^2 + (\beta_1+1) e^{\beta_2d_{h1}}.\label{eq18}
	\end{equation}
	By substituting $d_{h1}$, we can obtain the proposition 1.
\end{proof}

\subsubsection{The impact of node deployment quantity on navigation performance}

To explore the relationship between node deployment quantity and average navigation error, we will discuss the scaling law in this part. As Fig.~\ref{fig02} shows, the AUV sails horizontally from the left side to the right with constant velocity of $\boldsymbol{v}$ and positioning time slot $T$. To optimize the positioning results, we assume that the AUV passes through the center point of each anchor cluster and approximates the expected positioning error at each point within the cluster as the optimal value that $\mathbb{E}[\rm{CRLB}] = \frac{4\sigma_d^2}{N_{ca}\cos^2\alpha}+\frac{\sigma_d^2}{N_{ca}\sin^2\alpha}$, where $\sigma^2_{d} = \sum_{i=2}^I \frac{\gamma^2 s_{i-1}^2}{s_1^2 - s_{i-1}^2\cos^2\alpha}$ is the variation of LOS measurement error. Actually, this is the lower bound of expected navigation error. According to Fig.~\ref{fig02}, we have 
\begin{proposition}
	Scaling law, the expected navigation error is
	\begin{equation}
		\begin{aligned}
			P(N_{ta},N_{ca})  &= \frac{2d_{com}\lfloor \sqrt{N_c} \rfloor}{\sqrt{k}}\left(\frac{4\sigma_d^2}{N_{ca}\cos^2\alpha}+\frac{\sigma_d^2}{N_{ca}\sin^2\alpha}\right)\\
			+ &\frac{d_{h1}\left(\lfloor \sqrt{N_c} \rfloor+1\right)}{\sqrt{k}}\frac{1}{N} \sum_{n=1}^{N} \left(\sigma^2_{0,x} +\beta_1\mathrm{e}^{\beta_2\Delta p_{nT,x}}\right),
		\end{aligned}\label{eq19}
	\end{equation}
	where $N = \frac{d_{h1}}{\boldsymbol{v}T}$ is the number of sampling points used for calculating positioning errors.
\end{proposition}

\begin{proof}
	For AB (or CD, EF) section, the expected navigation error is:
	\begin{equation}
		\begin{aligned}
			P_{CD}(N_{ta},N_{ca}) &= \mathbb{E}_{CD}\left[tr(\boldsymbol{C}_{nv})\right]=\frac{1}{N} \sum_{n=1}^{N} \sigma^2_{nT,x}\\
			&= \frac{1}{N} \sum_{n=1}^{N} \left(\sigma^2_{0,x} +\beta_1\mathrm{e}^{\beta_2\Delta p_{nT,x}}\right).
		\end{aligned}\label{eq20}
	\end{equation}
	
	For BC (or DE) section, the expected positioning error representing by the center point is:
	\begin{equation}
		\begin{aligned}
			P_{BC}(N_{ta},N_{ca}) =\frac{4\sigma_d^2}{N_{ca}\cos^2\alpha}+\frac{\sigma_d^2}{N_{ca}\sin^2\alpha}.
		\end{aligned}\label{eq21}
	\end{equation}
	
	With total deployed anchors $N_{ta}$ and anchors in each cluster $N_{ca}$, the anchor clusters in X direction will be $\lfloor \sqrt{N_c} \rfloor = \lfloor \sqrt{\lfloor N_{ta}/N_{ca}\rfloor} \rfloor$. Then, there will be $\lfloor \sqrt{N_c} \rfloor +1$ AB (or CD, EF) sections and $\lfloor \sqrt{N_c} \rfloor $ BC (or DE) sections, so dividing based on the proportion of path length, the proportion of AB segment to the total length is $\frac{d_{h1}\left(\lfloor \sqrt{N_c} \rfloor+1\right)}{\sqrt{k}}$, and the proportion of BC segment to the total length is $\frac{2d_{com}\lfloor \sqrt{N_c} \rfloor}{\sqrt{k}}$. Combining these two coefficients with \eqref{eq20} and \eqref{eq21}, we can obtain \eqref{eq19}.
\end{proof}


\section{Results and Discussions}\label{sec4}
\subsection{Platform and Parameter settings}
In our simulations, the maximum depth of the region was set to be 3000 meters, and the optimal deployment grazing angle of anchors within the anchor group relative to the center point was 46 degrees according to \cite{Huang2024Opt}. The searching step of the traversal method for calculating equation \eqref{eq10} was set to be 100 meters. The communication coverage of anchor nodes was 5000 meters, and the error factor of time of arrival measurement was $\gamma = 0.001$, which means the standard deviation of error is one thousandth of the distance \cite{Huang2024Opt}. The AUV sails at a speed of 2 m/s, and the $\sigma_0$ is set to 0.1 m. 
\subsection{Accuracy performance}
To visually demonstrate the localization performance of anchor clusters, Fig.~\ref{fig05} gives an error distribution example with 3 anchors in each cluster at the depth of 500 meters. As the number of nodes within the cluster increases, the average CRLB and RMSE both increase, although the optimal CRLB of the center decreases \cite{Huang2024Opt}. This is mainly due to the poor geometric topology of the positioning service edge, which reduces the average positioning accuracy.


\vspace{-5mm}
\begin{figure}[!htbp]
	\centering
	\begin{minipage}[h]{0.24\textwidth}
		\centering
		\includegraphics[width=\linewidth]{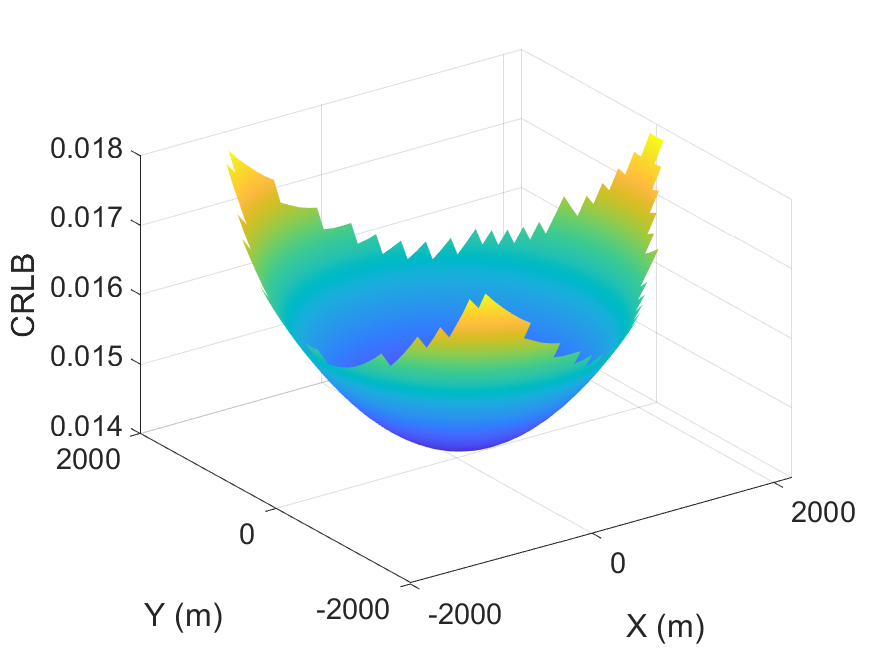}%
		\caption{Distribution of CRLB at depth of 500 meters with 3 anchors in each cluster.}
		\label{fig05}
	\end{minipage}
	\begin{minipage}[h]{0.24\textwidth}
		\centering
		\includegraphics[width=\linewidth]{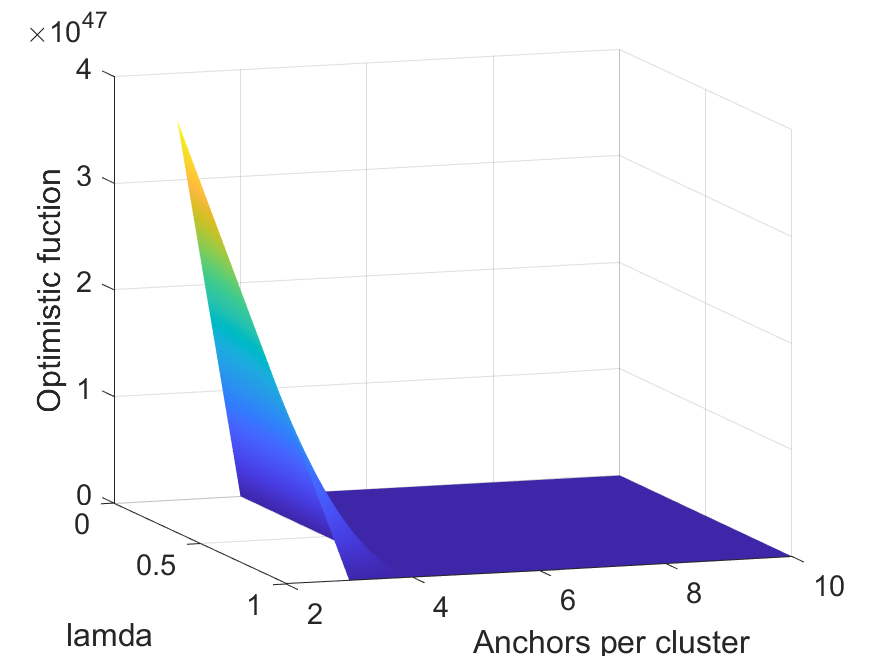}%
		\caption{Optimal value of problem \eqref{eq12} with different $\lambda_1$ and anchors per cluster.}
		\label{fig06}
	\end{minipage}
\end{figure}
\vspace{-5mm}

To find the optimal solution for equation \eqref{eq12}, we traversed different values of lambda and the sailing distance outside the anchor cluster. The sailing depth of AUV was assumed to be 500 meters, and navigation error coefficients were set to be $\beta_1 = 0.039$, $\beta_2 = 0.053$ based on the fitting of Fig.~\ref{fig03}a. The simulation result when the AUV only moved in the horizontal X or Y direction was given in Fig.~\ref{fig06}. 

\vspace{-5mm}
\begin{figure}[!htbp]
	\centering
	\begin{minipage}[h]{0.22\textwidth}
		\centering
		\includegraphics[width=\linewidth]{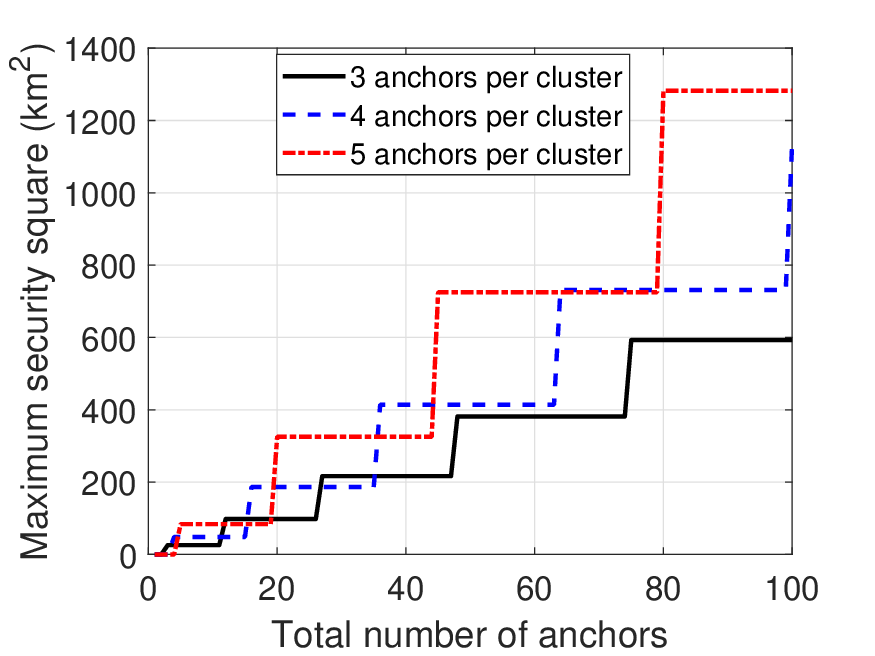}%
		\caption{Coverage square versus number of anchors.}
		\label{fig07}
	\end{minipage}
	\begin{minipage}[h]{0.26\textwidth}
		\centering
		\includegraphics[width=\linewidth]{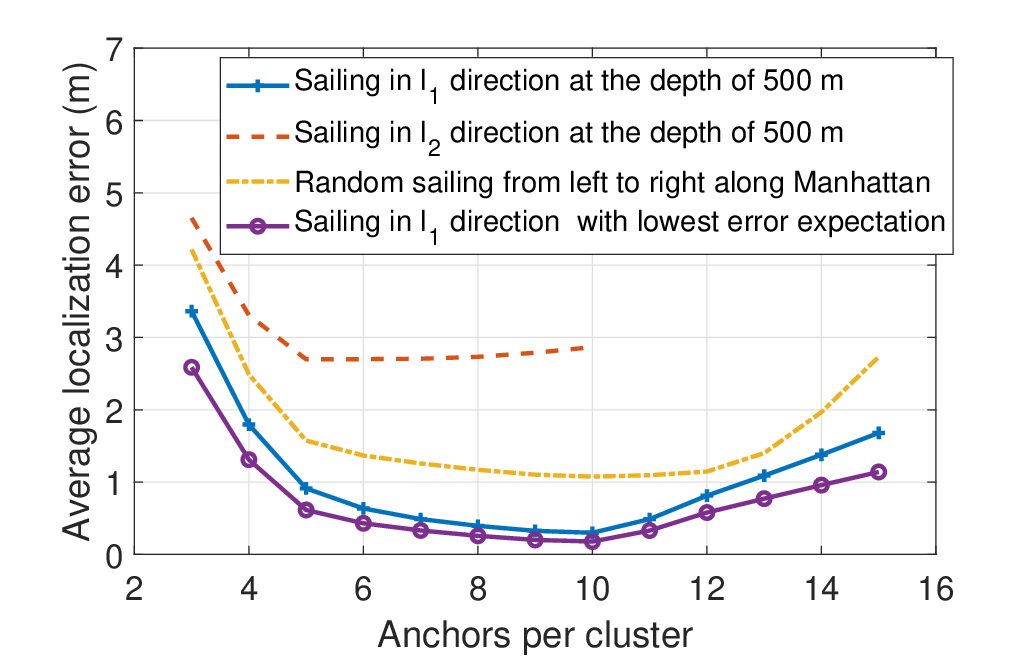}%
		\caption{Average localization error by different sailing paths.}
		\label{fig08}
	\end{minipage}
\end{figure}
\vspace{-5mm}

Due to the initial error variance of the inertial navigation system being higher than the expected CRLB of the anchor positioning system, the result in Fig.~\ref{fig06} shows that the $\lambda_1$ should be 1, which means seamless connection of the anchor clusters. Therefore, in the case of limited nodes, the optimal deployment scheme is to deploy proper nodes per cluster, maintaining the maximum coverage of anchor clusters. However, it is difficult to achieve seamless coverage of anchor clusters within a given large-scale area. At this point, to ensure that the AUV can smoothly navigate to nearby areas, there will be a upper bound on the coverage square of area as described by Proposition 1. Fig.\ref{fig07} shows the relations between coverage square and the maximum total number of anchors to be deployed. For any anchor cluster, the positioning service radius is 2279 meters for 3 anchors, 3213 meters for 4 anchors, and 4305 meters for 5 anchors.

To evaluate the impact of the number of nodes in different clusters on localization performance given the area and the total deployment node limit, we tested different paths (the random navigation curve adopts Monte Carlo simulation, and the starting and destination coordinates are randomly initialized) and presented the results in Fig.~\ref{fig08} with tests of 100 times. Fig.~\ref{fig08} is reasonable and consistent with the analysis of $\lambda$ mentioned earlier. For different path modes, the overall distribution interval is proportional to the length of the pure navigation path. Due to the fact that pure navigation error is much greater than the positioning error of long baseline positioning systems, the longer the pure navigation path, the greater the overall positioning error. For a certain kind of path, the deployment of nodes within different clusters results in different coverage areas for cluster positioning, but still follows the rule that the shorter the pure navigation spacing, the smaller the positioning error.


\section{Conclusion}\label{sec5}
In this paper, we established an optimization model that jointly considers the localization performance within and outside anchor clusters to minimize the overall navigation error. In the scenario of anchor positioning, we derived the CRLB under the effect of sound ray refraction. While in the scenario of navigation, we found the service area coverage conditions that ensure a high probability of reaching the destination and derived a scaling law about the influence of anchors in each cluster on the navigation performance within a given area. The simulation results confirmed our hypothesis.
\bibliographystyle{unsrt}
\bibliography{IEEEabrv,draft.bib}

\end{document}